\documentclass{article}

\usepackage{microtype}
\usepackage{graphicx}
\usepackage{subfigure}
\usepackage{booktabs} %

\usepackage{hyperref}

 \usepackage[accepted]{icml2025}

\usepackage{amsmath}
\usepackage{amssymb}
\usepackage{mathtools}
\usepackage{amsthm}

\usepackage[capitalize,noabbrev]{cleveref}

\theoremstyle{plain}
\newtheorem{theorem}{Theorem}[section]

\newtheorem{corollary}[theorem]{Corollary}
\theoremstyle{definition}
\newtheorem{definition}[theorem]{Definition}

\theoremstyle{remark}

\DeclareMathOperator*{\argmax}{arg\,max}

\usepackage[textsize=tiny]{todonotes}

\usepackage{color}
\usepackage{booktabs}
\usepackage{multirow}
\usepackage{adjustbox}
\usepackage{threeparttable}
\usepackage{comment}
\usepackage{wrapfig}
\usepackage{pifont}
\usepackage{enumitem}
\usepackage[hang,flushmargin]{footmisc}
\usepackage{bbm}

\usepackage{caption}
\DeclareCaptionLabelFormat{alglabel}{\bfseries Pseudocode}
\usepackage{balance}

\newcommand{\runningtitle}{Certified Robustness to Clean-Label Poisoning Using Diffusion Denoising}
\newcommand{\maintitle}{%
    Certified Robustness to Clean-Label Poisoning Using Diffusion Denoising}

\icmltitlerunning{\runningtitle{}}

\begin{document}

\twocolumn[
\icmltitle{\maintitle{}}

\icmlsetsymbol{equal}{*}

\begin{icmlauthorlist}
\icmlauthor{Sanghyun Hong}{osu}
\icmlauthor{Nicholas Carlini}{google}
\icmlauthor{Alexey Kurakin}{google}
\end{icmlauthorlist}

\icmlaffiliation{osu}{Oregon State University, Corvallis OR, USA}
\icmlaffiliation{google}{Google DeepMind, USA}

\icmlcorrespondingauthor{Sanghyun Hong}{sanghyun.hong@oregonstate.edu}

\icmlkeywords{Machine Learning, ICML}

\vskip 0.3in
]

\begin{abstract}
We present a certified defense
to clean-label poisoning attacks
under the $\ell_2$-norm.
These attacks work by injecting 
a small number of poisoning samples (e.g., 1\%)
that contain %
bounded adversarial perturbations 
into the training data 
to induce a targeted misclassification of a test input.
Inspired by the adversarial robustness
achieved by \emph{randomized smoothing},
we show how an %
off-the-shelf diffusion denoising model
can sanitize the tampered training data. %
We extensively test our defense 
against seven clean-label poisoning attacks
in both $\ell_2$ and $\ell_{\infty}$-norms
and reduce their attack success to 0--16\%
with only a %
negligible drop in the test accuracy.
We compare our defense
with existing countermeasures against clean-label poisoning,
showing that the defense reduces the attack success the most
and offers the best model utility.
Our results highlight the need for future work
on developing stronger clean-label attacks
and using our certified yet practical defense
as a strong baseline to evaluate these attacks. 
\end{abstract}

\section{Introduction}
\label{sec:intro}

A common practice in machine learning 
is to train models on a large corpus of data.
This paradigm empowers many over-parameterized deep-learning models
but also makes it challenging to retain the quality of data %
collected from various %
sources.
This makes deep-learning-enabled systems
vulnerable to \emph{data poisoning}~\citep{Antidote:09, SpamFilter:08, PoisonSVM:12, PoisonRegression:17, Frogs, PoisonSemiSupervise:21}%
---where an adversary can %
alter a victim model's behaviors
by injecting %
malicious samples (i.e., poisoning samples) into the training data.

But in practice, it may be challenging to make modifications to
entire training data used for %
supervised classification
and may be easy to detect label corruptions.
As a result, recent work has developed
\emph{clean-label poisoning attacks}~\citep{Frogs, CP, BP, LC, HT, MetaPoison, WitchesBrew},
where the attacker aims to control the victim model's behavior
on a specific test input by injecting a few poisons
that visually appear to be correctly-labeled,
but in fact include human-invisible adversarial perturbations.

We present a certified defense against %
clean-label poisoning attacks in $\ell_2$-norm
inspired by defenses to %
adversarial examples.
Existing poisoning defenses fall into two categories
\emph{certified} and \emph{heuristic}.
Certified defenses offer provable guarantees,
but often significantly decrease
the utility of defended models at inference,
making them impractical~\citep{DP:19, DPA, DPA2, BagFlip}.
Heuristic approaches~\citep{tRONI:18, DeepKNN:20, DP:20, AP:21, FrieNDs:22}
demonstrate their effectiveness 
against existing attacks in realistic scenarios. 
However, these defenses rely on unrealistic assumptions, 
such as the defender knowing the target test input~\citep{tRONI:18}, 
or are evaluated against specific poisoning adversaries, 
leaving them %
ineffective against adaptive attacks or future adversaries.

\textbf{Our contributions.}
\emph{First}, we make two seemingly distant goals closer:
presenting 
a certified defense against clean-label poisoning
that \emph{also} minimizes the decrease in clean accuracy.
For $\ell_2$-norm bounded adversarial perturbations to the training data,
we either achieve a certified accuracy higher than %
existing ones or require less computations to attain a comparable certified accuracy.
The model trained on the tampered data
classifies a subset of test input $x$ 
(or $x+\delta$, where $||\delta||_{\ell_2}$ 
in clean-label backdoor poisoning) correctly.
To achieve this goal,
we leverage the recent denoising diffusion probabilistic models%
~\citep{sohl2015deep, ho2020denoising, nichol2021improved}.
We use an off-the-shelf diffusion model to %
smooth the entire data before training a model
and provide certified predictions.

Removing adversarial perturbations in the training data
before model training, we can decouple 
the certification process from the training algorithm.
\emph{Second}, we leverage this computational property %
our defense has
and present a series of training techniques 
to alleviate the side-effects of employing certified defenses,
i.e., the utility loss of a defended model.
To our knowledge,
we are the first %
certified defense that decouples these two processes.
We minimize the utility loss by employing the warm-starting~\citep{ash2020warm}.
We train a model a few epochs on the tampered data or 
initialize its model parameters using a pre-trained model.
None of them alters the training process;
thus, our defense's provable guarantee holds.
In \S\ref{sec:evaluation},
we show they attribute to a better performance
compared to standard training.

\emph{Third}, we extensively evaluate our defense 
against seven clean-label poisoning attacks studied in two different scenarios:
transfer-learning and training a model from-scratch.
In transfer-learning, 
our defense completely renders the attack ineffective
with a negligible accuracy drop of 0.5\% 
in the certified radius of 0.1 in $\ell_2$-norm perturbations.
We also reduce the attack success to 2--16\%
when a defender trains a model from scratch.
We further compare our defense
with five poisoning defenses in the prior work.
We demonstrate more (or the same in a few cases)
reduction in the attack success 
and less accuracy drop than those existing defenses.

\section{Preliminaries on Clean-label Poisoning}
\label{sec:prelim}

\noindent
\textbf{Threat model.}
A clean-label poisoning attacker causes %
a misclassification of a specific \emph{target} test sample (${x}_{t}$, $y_t$) 
by compromising the training data $D$ with %
poisoning samples $D_{p}$. 
If a victim trains a model $f$ on the poisoned training data $D'\!=\!D \cup D_{p}$,
the resulting model $f^{*}_{\theta}$ is likely to misclassify 
the target instance to the adversarial class $y_{adv}$ 
while preserving the classification behavior of $f^{*}_{\theta}$ %
on the clean test-set $S$.
The attacker crafts those \emph{poisons} $({x}_{p}, y_p)$ by first taking a few \emph{base} samples in the same domain (${x}_{b}, y_b$) and then adding human-imperceptible perturbations $\delta$, carefully crafted by the attacker and also bound to $||\delta||_{\ell_{p}} \leq \epsilon$, to them while keeping their labels %
\emph{clean} ($y_b=y_p$).

\noindent
\textbf{Poisoning as a constrained bi-level optimization.}
The process of crafting optimal poisoning samples $D^{*}_p$
can be formulated as the constrained bi-level optimization problem:
\begin{align*}
    D^{*}_p = \underset{D_p}{\arg\min}\text{ }
    \mathcal{L}_{adv}(x_t, y_t; f^{*}_{\theta}),
\end{align*}
where $\mathcal{L}_{adv}(x_t, y_{adv}; f^{*}_{\theta})$ is the adversarial loss function quantifying how accurately a model $f^{*}_{\theta}$, trained on the compromised training data, misclassifies a target sample $x_t$ into the class an adversary wants $y_{adv}$. $D_p$ is the set of poisoning samples we craft, and $D^{*}_p$ is the resulting optimal poisons.

While minimizing the crafting objective, %
the attacker also trains a model $f^{*}_{\theta}$ on the tampered training data, 
which is itself another optimization problem, formulated as follows:
\begin{align*}
    f^{*}_{\theta} = \underset{\theta}{\arg\min}\text{ }
    \mathcal{L}_{tr}(D_{ptr}, S; \theta),
\end{align*}
where the typical choice of $\mathcal{L}_{tr}$ is the cross-entropy loss, and $S$ is the clean test-set. Combining both the equations becomes a bi-level optimization: find $D_p^*$ such that $\mathcal{L}_{adv}$ is minimized after training, while minimizing $\mathcal{L}_{tr}$ as well.

To make the attack inconspicuous,
the attacker \emph{constraints} this bi-level optimization
by limiting the %
perturbation $\delta = x_p - x_b$
each poisoning sample can contain to $||\delta||_{\ell_p} < \epsilon$.

\noindent
\textbf{Existing clean-label attacks.}
Initial work~\citep{Frogs, CP, BP} minimizes $\mathcal{L}_{adv}$ by crafting %
poisons that are close to the target in the latent representation space $g(\cdot)$. A typical choice of $g(\cdot)$ is the activation outputs from the penultimate layer of a pre-trained model $f(\cdot)$. 
The attacks have shown effective in \emph{transfer-learning} scenarios, where $f(\cdot)$ will be fine-tuned on the tampered training data. The attacker chooses base samples from the target class $(x_b, y_{adv})$ and craft poisons $(x_p, y_{adv})$. During fine-tuning, $f(\cdot)$ learns to correctly classify poisons in the target's proximity in the latent representation space %
and classify the target into the class $y_{adv}$ the attacker wants.

Recent work %
make those poisoning attacks effective 
when $f$ is trained \emph{from scratch} on %
$D'$. 
To do so, 
the attacker requires to approximate the gradients computed on $\mathcal{L}_{adv}$ that are likely to appear in any models. \citet{MetaPoison} %
employs meta-learning; the poison-crafting process simulates all the possible initialization, intermediate models, and adversarial losses computed on those intermediate models. %
\citet{WitchesBrew} reduces this computational overhead %
by proposing %
gradient matching, that aligns the gradients from poisoning samples with those computed on a target.

\noindent
\textbf{Defenses against clean-label poisoning.} 
Early work on poisoning defenses~\cite{RONI, tRONI:18} focuses on filtering out poisons from the training data.
Follow-up work~\citep{Antidote:09, DeepKNN:20, Spectral} leverages unique statistical properties (\textit{e.g.}, spectral signatures) that can distinguish poisons from the clean data. All these approaches depend on the data they use to compute the properties. Recent work~\citep{AP:21, DP-InstaHide:21, DP:20, FrieNDs:22} thus reduces the dependency by proposing data- or model-agnostic defenses, adapting robust training or differentially-private training. %
While shown effective against existing attacks, those defenses are not \emph{provable}.

A separate line of work %
proposes
\emph{certified} defenses~\citep{DPA, DPA2, RAB}, which guarantee the correct classification of test samples when an adversary compromises the training data. %
\citet{DPA, DPA2}'s approach %
splits the training data into multiple disjoint subsets, trains models on them, and runs majority voting of test inputs over these models. 
BagFlip~\cite{BagFlip} proposes 
a model-agnostic defense %
based on randomized smoothing~\cite{cohen2019certified}, studied in adversarial robustness.
But, 
the computational demands hinder their deployment in practice.

\section{Diffusion Denoising as Certified Defense}
\label{sec:certified-defense}

\subsection{Robustness Guarantee}
\label{subsec:certification-proof}

We formally show that our defense provides 
a provable robustness guarantee against clean-label poisoning attacks,
particularly under $\ell_2$-norm perturbations.
Due to space constraints,
we only include important theorems and corollaries here
and provide the full details in Appendix~\ref{appendix:certification-proof}.

\citet{cohen2019certified} shows that 
a smoothed classifier $g$, constructed via \emph{randomized smoothing},
is robust around a test input $x$ within a certified radius of $R$.
The following theorem from the original study 
establishes this connection:

\begin{theorem}[Robustness guarantee by~\citet{cohen2019certified}]
\label{thm:main-cohen}
    Let $f: \mathbb{R}^{d}\rightarrow \mathcal{Y}$ be any deterministic (or random) function,
    $\varepsilon \sim \mathcal{N}(0, \sigma^2, I)$, and
    $g$ to be the smoothed classifier robust under $\varepsilon$.
    Suppose $y_A\in\mathcal{Y}$ and $p_A, p_B \in [0, 1]$ satisfy:
    \begin{align*}
        \mathbb{P}(f(x+\varepsilon)=y_A) \geq p_A > p_B
        \geq \underset{y \neq y_A}{\max}\mathbb{P}(f(x+\varepsilon)=c)
    \end{align*}
    then $g(x+\delta)=y_A$ for all $||\delta||_2 < R$, where
    \begin{align*}
        R = \frac{\sigma}{2}\big(\Phi^{-1}(p_A) - \Phi^{-1}(p_B)\big)
    \end{align*}
    and $\Phi$ is the inverse of the standard Gaussian CDF.
\end{theorem}

Theorem~\ref{thm:main-cohen} does not impose specific restrictions on 
\emph{how} the smoothed classifier $g$ is constructed.
A common approach prior work employs is to 
apply additive Gaussian noise to the training data
or to the internal activations~\cite{lecuyer2019certified}
\emph{at each mini-batch during training}.
But this approach is not compatible with our setting
where we %
want to \emph{smooth the training data}.
To address this challenge, 
we adopt the approach by~\citet{RAB},
which obtains a smoothed classifier $g$
through a method that aligns with our setting:

\begin{definition}[$g$ considered in~\citet{RAB}]
\label{def:main-smoothed-classifier-rab}
    Suppose a base classifier $h$, learned from the training data $D$
    and to maximize the conditional probability distribution over labels $y$.
    Then the smoothed classifier $q$ is defined by:
    \begin{align*}
        q(y|x, D) = \mathbb{P}_{X, D}(h(x+X, D+\mathcal{D})=y),
    \end{align*}
    where $X\sim\mathbb{P}_X$ is random variables,
    and $\mathcal{D}\sim\mathbb{P}_\mathcal{D}$ denotes smoothing distributions.
    $\mathcal{D}$ is a collection of $n$ i.i.d random variables,
    composed of $\mathcal{D}^{i}$s, where each $\mathcal{D}^{i}$ 
    will be added to a training samples in $D$.
    $\mathbb{P}_X$ is the noise distribution for test inputs,
    and $\mathbb{P}_\mathcal{D}$ is the noise distribution for the training data.
    The final, smoothed classifier then assigns 
    the most likely class to an instance $x$ under $q$, so that:
    \begin{align*}
        g(x, D) = \argmax_{y} q(y|x, D)
    \end{align*}
\end{definition}

Our work considers the smoothed classifier $g$ 
defined in~\ref{def:main-smoothed-classifier-rab}.
We then establish our robustness guarantee,
based on the guarantee against backdoor attacks,
which is completely different from clean-label poisoning,
as follows:

\begin{figure*}[ht]
\centering
\vspace{-1.0em}
\begin{minipage}[t]{0.48\textwidth}
\begin{algorithm}[H]
    \caption{Noise, denoise, train, and classify}
    \label{algo:noise-denoise-classify}
\begin{algorithmic}[1]
    \STATE \textbf{fn} {\small\textsc{NDTClassify}}($f, \sigma, D, x, n$)
    \STATE \text{  } \verb!counts! $\gets \mathbf{0}$
    \STATE \text{  } \textbf{for} $i \in \{1, 2, ..., n\}$ \textbf{do}
    \STATE \text{  } \text{  } $t^{*}, \alpha_{t^{*}} \gets$ {\small\textsc{GetTimeStep}}($\sigma$)
    \STATE \text{  } \text{  } $\hat{D} \gets$ {\small\textsc{NoiseAndDenoise}}($D, \alpha_{t^{*}}; t^{*}$)
    \STATE \text{  } \text{  } $\hat{f_\theta} \gets$ {\small\textsc{Train}}($\hat{D}, f$)
    \STATE \text{  } \text{  } \verb!counts![$\hat{f_\theta}(x)$] $\gets$ \verb!counts![$\hat{f_\theta}(x)$] + 1
    \STATE \text{  } \textbf{ret} \verb!counts!
    \STATE
    \STATE \textbf{fn} {\small\textsc{GetTimeStep}}($\sigma$)
    \STATE \text{  } $t^{*} \gets$ find $t$ s.t. $\frac{1-\alpha_t}{\alpha_t} = \sigma^2$
    \STATE \text{  } \textbf{ret} $t^{*}, \alpha_{t^{*}}$
    \vspace{2.2em}      %
\end{algorithmic}
\end{algorithm}
\end{minipage}
\hfill
\begin{minipage}[t]{0.50\textwidth}
\begin{algorithm}[H]
    \caption{Predict and certify~\citep{cohen2019certified}}
    \label{algo:prediction-n-certification}
\begin{algorithmic}[1]
    \STATE \textbf{fn} \textsc{Predict}($f, \sigma,$ {\color{blue}$D$}$, x, m, \alpha$)
    \STATE \text{  } \verb!counts! $\gets$ {\color{blue}{\small\textsc{NDTClassify}($f, \sigma, D, x, m_0$)}}
    \STATE \text{  } $\hat{c_A}, \hat{c_B} \gets$ top two predictions in \verb!counts!
    \STATE \text{  } $n_A, n_b \gets$ \verb!counts![$\hat{c_A}$], \verb!counts![$\hat{c_B}$]
    \STATE \text{  } \textbf{if} {\small\textsc{BinomPVal}($n_A, n_A + n_B, 0.5 \leq \alpha$)} \textbf{ret} $\hat{c_A}$
    \STATE \text{  } \textbf{else ret} {\small\textsc{ABSTAIN}} 
    \STATE
    \STATE \textbf{fn} \textsc{Certify}($f, \sigma,$ {\color{blue}$D$}$, x, m_0, m, \alpha$)
    \STATE \text{  } \verb!counts0! $\gets$ {\color{blue}{\small\textsc{NDTClassify}($f, \sigma, D, x, m_0$)}}
    \STATE \text{  } $\hat{c_A} \gets$ top predictions in \verb!counts0!
    \STATE \text{  } \verb!counts! $\gets$ {\color{blue}{\small\textsc{NDTClassify}($f, \sigma, D, x, m$)}}
    \STATE \text{  } $p_A \gets$ {\small\textsc{LowerCfBound}(\verb!counts![$\hat{c_A}$], $m$, $1-\alpha$)}
    \STATE \text{  } \textbf{if} $p_A > 1/2$ \textbf{ret} $\hat{c_A}$ and radius $\sigma\Phi^{-1}(p_A)$
    \STATE \text{  } \textbf{else ret} {\small\textsc{ABSTAIN}}
\end{algorithmic}
\end{algorithm}
\end{minipage}
\label{lbl:pseudocode}
\vspace{-1.0em}
\end{figure*}

\begin{theorem}[Robustness guarantee by~\citet{RAB}]
\label{thm:main-rab}
    Suppose $g$ is the smoothed classifier defined above
    with smoothing distribution $Z:=(X, \mathcal{D})$ with $X$ taking values in $\mathbb{R}^{d}$
    and $\mathcal{D}$ being a collection of $n$ independent $\mathbb{R}^{d}$-valued random variables,
    $\mathcal{D}=(\mathcal{D}^1, ..., \mathcal{D}^n)$.
    Let $\Omega_x \in \mathbb{R}^d$ and 
    $\Delta := (\delta^1, \delta^2, ..., \delta^n)$ 
    for backdoor patterns $\delta^i \in \mathbb{R}^d$.
    Let $y_A \in \mathcal{Y}$ and $p_A, p_B \in [0, 1]$
    such that $y_A = g(x, D)$ and
    \begin{align*}
        q(y_A|x, D) \geq p_A > p_B \geq 
        \underset{y \neq y_A}{\max}q(y|x, D).
    \end{align*}
    If the optimal type II errors, for testing the null $Z \sim \mathbb{P}_0$
    against the alternative $Z + (\Omega_x, \Delta) \sim \mathbb{P}_1$, satisfy
    \begin{align*}
        \beta^{*}(1-p_Z; \mathbb{P}_0, \mathbb{P}_1) +
        \beta^{*}(p_B; \mathbb{P}_0, \mathbb{P}_1) > 1,
    \end{align*}
    then $y_A = \argmax_{y} q(y|x+\Omega_x, D+\Delta)$.
\end{theorem}

Theorem~\ref{thm:main-rab} establishes that the smoothed classifier $q$,
as defined in~\ref{def:main-smoothed-classifier-rab},
trained on data containing perturbations $\Delta$,
is robust around a test input $x$
with the perturbation $\Omega_x$.

\begin{corollary}[Certified radius by~\citet{RAB}]
\label{coll:main-l2-rab}
    Let $\Delta = (\delta_1, \delta_2, ..., \delta_n$),
    $\Omega_x$ be $\mathbb{R}^d$-valued backdoor patterns,
    and $D$ be a training dataset.
    Suppose that for each $i$,
    the smoothing noise on the training features is 
    $\mathcal{D}^i \stackrel{iid}\sim \mathcal{N}(0, \sigma^2 \mathbbm{1}_d)$.
    Let $y_A \in \mathcal{Y}$ such that $y_A=q(x+\Omega_x, D+\Delta)$
    with class probabilities satisfying:
    \begin{align*}
        q(y_A|x+\Omega_x, D+\Delta) & \geq p_A \\
        & > p_B \geq \underset{y \neq y_A}{\max} q(y|x+\Omega_x, D+\Delta)
    \end{align*}
    If the perturbations on the training samples are bounded by
    \begin{align*}
        \sqrt{\sum_{i=1}^{n} ||\delta_i||_2^2} <
        \frac{\sigma}{2}\big(\Phi^{-1}(p_A) - \Phi^{-1}(p_B)\big),
    \end{align*}
    it is guaranteed $y_A = q(x+\Omega_x, D) = q(x+\Omega_x, D+\Delta)$.
\end{corollary}

Now we can simplify Corollary~\ref{coll:main-l2-rab}
by assuming an adversary compromises $r \leq n$ training samples
with the largest noise pattern's $\ell_2$-norm bound, $||\delta||_2$:
\begin{align*}
    ||\delta||_2 < \frac{\sigma}{2\sqrt{r}}
        (\Phi^{-1}(p_A) - \Phi^{-1}(p_B)\big),
\end{align*}

A typical choice of $r$ in clean-label poisoning is $\sim$10\%.
But when computing our certificate bound in \S\ref{sec:certified-defense},
we consider the worst-case adversary 
capable of compromising the entire training set.
In this case, $r\!=\!n$.
Under this formulation, 
we observe that against the strongest adversary 
who can compromise the entire training data $n$,
the robustness bound is much smaller compared to
that against the realistic attacker capable of 
compromising only 10\% of the training data.

Our certification process trains $m$ models
that return the prediction of $x$, the target test sample, as $y_A$.
Then the lower bound of $p_A$ becomes $\alpha^{1/m}$,
with the probability of at least $1-\alpha$.
Plugging this into the above equation yields:
\begin{align*}
    ||\delta||_2 < \frac{\sigma}{\sqrt{m}} (\Phi^{-1}(p_A)\big)
\end{align*}

\begin{corollary}[Our robustness guarantee against clean-label poisoning.]
    Consider the smoothed classifier $q$ from Theorem~\ref{thm:main-rab},
    with no test input perturbation $X$ and 
    the smoothing mechanism being diffusion denoising,
    which applies a Gaussian process to add and remove noise from the input.
    Under this setting, it provides the same robustness guarantee 
    for clean-label poisoning attacks.
\end{corollary}

\begin{proof}
    By setting $X \equiv 0$ in Theorem~\ref{thm:main-rab} (and therefore $\Omega_x = 0$),
    the smoothed classifier will output the same label prediction $y_A$
    within the same $\ell_2$-bound for the perturbations applied to the training data of size $n$.
\end{proof}

\subsection{(Diffusion) Denoising for the Robustness}
\label{subsec:diffusion-denoising}

\noindent 
\textbf{Denoising diffusion probabilistic model (DDPMs)}
are a recent generative model
that works by learning the diffusion process of the form
$x_t \sim \sqrt{1 - \beta_t} \cdot x_{t-1} + \beta_t \cdot \omega_t$,
where $\omega_t$ is drawn from a standard normal distribution $\mathcal{N}(0, \mathbf{I})$
with $x_0$ sourced from the actual data distribution,
and $\beta_t$ being fixed (or learned) variance parameters.
This process transforms images
from the target data distribution into purely random noise over time $t$,
and the reverse \emph{denoising} process constructs images in the data distribution,
starting with random Gaussian noise.
A DDPM with a fixed time-step $t \in \mathbb{N}^{+}$ and a fixed schedule
samples a noisy version of a training image $x_t \in [-1, 1]^{w \cdot h \cdot c}$ of the form:
\begin{align*}
    x_t := \sqrt{\alpha_t} \cdot x + \sqrt{1 - \alpha_t} \cdot \mathcal{N}(0, \mathbf{I}),
\end{align*}
where $\alpha_t$ is a constant derived from $t$,
which decides the level of noise to be added to the image 
(the noise increases consistently as $t$ grows).
During training,
the model minimizes the difference between 
$x$ and the denoised %
$x_t$,
where $x_t$ is obtained by applying the noise at time-step $t$.

\noindent
\textbf{Diffusion denoising for the robustness.}
We utilize off-the-shelf
DDPMs~\citep{sohl2015deep, ho2020denoising, nichol2021improved}
to \emph{denoise} adversarial perturbations added to the training data
and as a result, provide the robustness to clean-label poisoning.
A naive adaptation of randomized smoothing to our scenario 
is to train multiple models on the training data 
with Gaussian noise augmentation or using adversarial training.
But we need to add noise to the data 
that does not make it look natural 
or they are computationally demanding.
We thus avoid using these approaches to train models, 
instead we want to remove the (potentially compromised) training set before training. %
We use a single-step denoising,
as shown in Pseudocode~\ref{algo:noise-denoise-classify}.

\begin{figure*}[ht]
\centering
\vspace{-0.4em}      %
\begin{minipage}{0.49\linewidth}
    \centering
    \includegraphics[width=0.8\linewidth]{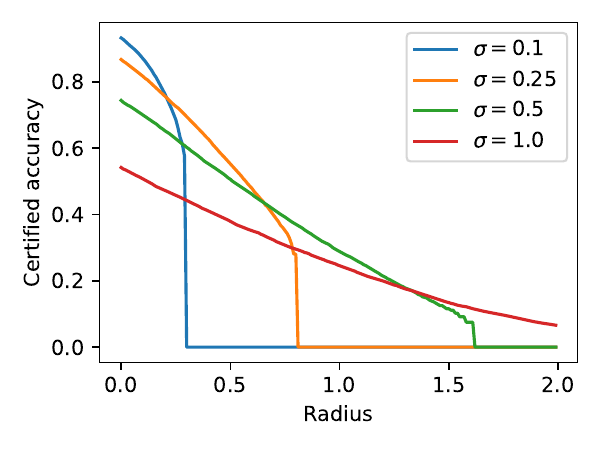}
\end{minipage}
\hfill
\begin{minipage}{0.49\linewidth}
    \centering
    \includegraphics[width=0.8\linewidth]{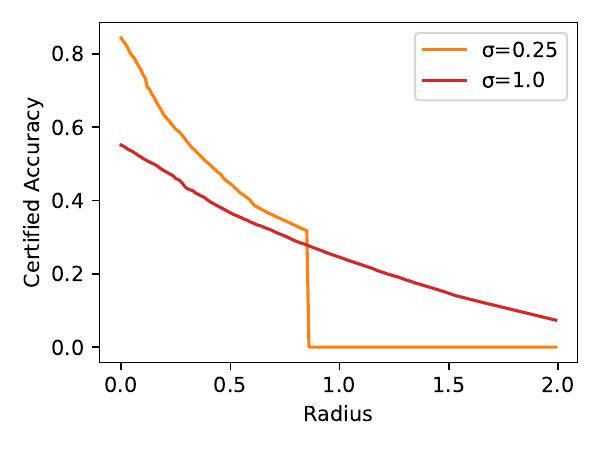}
\end{minipage}
\vspace{-1.4em}      %
\caption{\textbf{Certified radius and accuracy} attained by \emph{denoising} the CIFAR10 training data with varying $\sigma$ values in \{0.1, 0.25, 0.5, 1.0\} (left) and by adding \emph{Gaussian random noise} to the training data
with $\sigma$ values in \{0.25, 1.0\} (right).}
\label{fig:acc-vs-radius}
\vspace{-1.0em}
\end{figure*}

\subsection{Prediction and Certification}
\label{subsec:prediction-certification}

Now we present %
algorithms 
for running a prediction with the classifiers trained on $D_{ptr}$ 
and certifying the robustness of the model’s prediction 
on a test sample $x$.
Our work extends the algorithms presented by~\citet{cohen2019certified}
to clean-label poisoning settings, 
as shown in Pseudocode~\ref{algo:prediction-n-certification}.

\noindent
\textbf{Prediction.}
To compute the prediction for a test sample $x$,
we make an adaptation of the standard randomized smoothing.
We train $m$ classifiers on $m$ different noised-then-denoised training datasets,
and return the predictions of these classifiers on the target $x$,
with the algorithm {\small\textsc{DTClassify}}($\cdot$).
We run this algorithm for a sufficient number of times (e.g., over 1000 times).
The prediction output is then the majority-voted label from these classifiers.

\noindent
\textbf{Certification}
process exactly follows the standard randomized smoothing.
We first count the number of occurrences of 
the most likely label $\hat{y_A}$ compared to any other runner-up label,
and from this can derive a radius 
(for the training set perturbations)
on which $x$ is guaranteed to be robust.

\noindent
\textbf{Results.}
We train 10k CIFAR10 models using an efficient training pipeline that
trains a single model in $15$ seconds on an A100 GPU
(for a total of 1.7 GPU days of compute).
These models have a clean accuracy of between $92\%$ and $58\%$
depending on the level of noise $\sigma$ introduced, 
and for the first time can certify non-trivial robustness to
clean-label poisoning attacks. %
A full curve comparing certified accuracy for all given perturbation radii
is shown in Figure~\ref{fig:acc-vs-radius}.
Instead of applying a diffusion model,
we test the randomized smoothing approach proposed by~\citet{RAB}:
adding Gaussian random noise 
directly to the images and train on these (noisy) images.
This achieves a similar level of certified accuracy and robustness
as our denoising approach, but at a cost of training for $20\times$ longer.

\subsection{Computational Efficiency for Certification}
\label{subsec:efficiency}

In order to obtain a robustness certificate, 
we must train many CIFAR10 models
on different denoisings of the training dataset.
This step %
is computationally expensive,
but is in line with the computational complexity of prior 
randomized smoothing based approaches that require similar work.
However, as we will show later, in practice we can empirically obtain robustness
with only a very small number of training runs (even as low as \textbf{one}!),
confirming the related observation from~\citep{lucas2023randomness} %
who find a similar effect holds for randomized smoothing of classifiers.

\section{Empirical Evaluation}
\label{sec:evaluation}

We empirically evaluate the effectiveness of our certified 
defense against clean-label poisoning. We adapt the poisoning benchmark developed by~\citet{PoisoningBenchmark}.

\noindent
\textbf{Poisoning attacks.}
We evaluate our defense against seven clean-label poisoning attacks: four targeted poisoning (Poison Frogs~\citep{Frogs}, Convex Polytope~\citep{CP}, Bullseye Polytope~\citep{BP}, Witches' Brew~\citep{WitchesBrew}) and three backdoor attacks (Label-consistent Backdoor~\citep{LC}, Hidden-Trigger Backdoor~\citep{HT}, and Sleeper Agent~\citep{SleeperAgent}). The first three %
attacks operate in the transfer-learning scenarios (i.e., they assume that the victim fine-tunes a model from some given initialization), whereas the rest assumes the victim trains a model from-scratch.

Because our guarantee is a special case of the one provided by~\cite{RAB},
the certification holds for both backdooring and clean-label poisoning in $\ell_2$.

We evaluate both the transfer-learning and training from-scratch scenarios.
We also evaluate %
against the attackers with two types of knowledge: %
white-box and black-box. The white-box attackers know the victim model's %
internals, while the black-box attackers do not. 
More details are
in Appendix~\ref{appendix:detailed-configurations}.

\noindent
\textbf{Metrics.}
We employ two metrics: clean accuracy and attack success rate. 
We compute the accuracy on the entire test set. %
An attack is successful if a poisoned model classifies a target sample 
to the intended class. When evaluating backdoor attacks, 
we only count the cases where a tampered model misclassifies 
the target sample containing a trigger pattern.

\begin{table*}[h]
\centering
\caption{\textbf{Defense effectiveness in transfer-learning scenarios (CIFAR10).} We measure the clean accuracy and attack success of the models trained on the denoised training set. Each cell shows the accuracy in the parentheses and the attack success outside. Note that $^\dagger$ indicates the runs with $\sigma$=0.0, the same as our baseline that trains models without any denoising.} %
\label{tbl:cifar10-denoise-linf16-transfer-learning-all}
\vspace{-0.6em}
\adjustbox{max width=1.0\textwidth}{
\begin{tabular}{@{}lc | rrrrr | rrrrr@{}}
    \toprule
    \multicolumn{1}{c}{} &  
    & \multicolumn{5}{c}{\textbf{Our defense against $\ell_2$ attacks at $\sigma$ (\%)}} 
    & \multicolumn{5}{|c}{\textbf{Our defense against $\ell_{\infty}$ attacks at $\sigma$ (\%)}} \\ \cmidrule(l){3-12} 
    \textbf{Poisoning attacks} & \textbf{Knowledge} 
    & $^\dagger$\textbf{0.0} & \textbf{0.1} & \textbf{0.25} & \textbf{0.5} & \textbf{1.0} 
    & $^\dagger$\textbf{0.0} & \textbf{0.1} & \textbf{0.25} & \textbf{0.5} & \textbf{1.0}
    \\ \midrule \midrule
    Poison Frog! & \multirow{5}{*}{\rotatebox[origin=c]{90}{WB}} 
    & $^{(93.6)}$99.0 & $^{(93.3)}$0.0 & $^{(91.8)}$1.0 & $^{(84.8)}$0.0 & $^{(79.9)}$1.0 
    & $^{(93.6)}$68.8 & $^{(93.3)}$0.0 & $^{(92.7)}$0.0 & $^{(90.8)}$0.0 & $^{(87.4)}$0.0 \\
    Convex Polytope &  
    & $^{(93.7)}$16.2 & $^{(93.2)}$0.0 & $^{(91.7)}$0.0 & $^{(86.6)}$0.0 & $^{(77.0)}$0.0 
    & $^{(93.7)}$12.2 & $^{(93.3)}$0.0 & $^{(92.7)}$0.0 & $^{(90.8)}$1.0 & $^{(87.5)}$0.0 \\
    Bullseye Polytope &  
    & $^{(93.5)}$100 & $^{(93.3)}$4.0 & $^{(92.6)}$0.0 & $^{(87.5)}$0.0 & $^{(79.2)}$1.0
    & $^{(93.5)}$100 & $^{(93.3)}$0.0 & $^{(92.7)}$0.0 & $^{(90.8)}$1.0 & $^{(87.5)}$0.0 \\
    Label-consistent Backdoor &  
    & \multicolumn{5}{c|}{\textbf{-}}
    & $^{(93.2)}$1.0 & $^{(93.3)}$0.0 & $^{(92.6)}$0.0 & $^{(90.8)}$1.0 & $^{(87.5)}$0.0 \\ 
    Hidden Trigger Backdoor &  
    & \multicolumn{5}{c|}{\textbf{-}}
    & $^{(93.4)}$7.0 & $^{(93.3)}$0.0 & $^{(92.6)}$0.0 & $^{(90.8)}$0.0 & $^{(87.5)}$0.0
    \\ \midrule
    Poison Frog! & \multirow{5}{*}{\rotatebox[origin=c]{90}{BB}} 
    & $^{(91.6)}$10.0 & $^{(91.2)}$0.0 & $^{(89.6)}$0.5 & $^{(82.9)}$0.0 & $^{(77.8)}$2.0
    & $^{(91.7)}$2.5 & $^{(91.3)}$0.0 & $^{(90.3)}$0.0 & $^{(88.8)}$0.5 & $^{(86.2)}$1.0 \\
    Convex Polytope &  
    & $^{(91.7)}$3.0 & $^{(91.0)}$0.0 & $^{(89.5)}$0.0 & $^{(84.6)}$0.5 & $^{(73.6)}$1.0
    & $^{(91.8)}$2.5 & $^{(91.3)}$0.0 & $^{(90.3)}$0.0 & $^{(88.8)}$0.5 & $^{(86.2)}$1.0 \\
    Bullseye Polytope &  
    & $^{(91.6)}$9.0 & $^{(91.3)}$0.0 & $^{(90.3)}$0.0 & $^{(85.5)}$0.0 & $^{(76.3)}$1.0
    & $^{(91.6)}$8.0 & $^{(91.3)}$0.0 & $^{(90.3)}$0.0 & $^{(88.8)}$0.5 & $^{(86.2)}$0.5 \\ 
    Label-consistent Backdoor &  
    & \multicolumn{5}{c|}{\textbf{-}}
    & $^{(91.5)}$1.0 & $^{(91.3)}$0.0 & $^{(90.3)}$0.0 & $^{(88.8)}$0.0 & $^{(86.2)}$1.5 \\
    Hidden Trigger Backdoor &  
    & \multicolumn{5}{c|}{\textbf{-}}
    & $^{(91.6)}$4.0 & $^{(91.2)}$1.0 & $^{(90.3)}$1.0 & $^{(89.3)}$1.5 & $^{(86.3)}$1.5 \\ \bottomrule
\end{tabular}
}
\vspace{-1.em}      %
\end{table*}

\noindent
\textbf{Methodology.}
We run our experiments in CIFAR10 and Tiny ImageNet~\citep{CIFAR10, TinyImageNet}.
Following the standard practices in the prior work~\citep{PoisoningBenchmark}, 
we run each attack 100 times on different targets and report the averaged metrics over the 100 runs.
We randomly choose the base and target classes.
We set the poisoning budget to 0.5--1\% for all our attacks:
the attacker can tamper with only 500 samples of the entire 50k training instances in CIFAR10
and 250 of the total 100k training samples in Tiny ImageNet.

\subsection{Effectiveness of Our Denoising Defense}
\label{subsec:our-defense-eval}

We first show to what extent our defense makes the seven poisoning attacks ineffective. Our defense is attack-agnostic. A defender does not know whether the training data contains poisoning samples or not; instead, the defender always trains a model on the denoised training set. We measure the trained model's accuracy and attack success and compare them with the baseline, where a model is trained directly without denoising. We evaluate with different defense strengths: $\sigma \in \{0.1, 0.25, 0.5, 1.0\}$. %
The attacks are with the perturbation bound of 16-pixels in $\ell_{\infty}$-norm and that of 3.5--6.5 in $\ell_{2}$-norm. We extended many poisoning attacks for supporting $\ell_2$ bounds. Both $\ell_{\infty}$ and $\ell_2$ bounds we use are comparable to each other. We show the conversion between the two bounds in Appendix~\ref{appendix:norm-conversion}. 
Due to the page limit, %
the results from Tiny ImageNet %
are in Appendix~\ref{appendix:additional-experiments}.

\noindent
\textbf{Transfer-learning scenarios.}
Table~\ref{tbl:cifar10-denoise-linf16-transfer-learning-all} summarizes our results. Our defense is also effective in mitigating the attacks in transfer-learning. Each cell contains the accuracy and the attack success rate of a model trained on the denoised training set. All the numbers are the averaged values over 100 attacks. Against the white-box attacks, we reduce the success rate of the targeted poisoning attacks to 0--4\% at $\sigma$ of 0.1. In the black-box attacks, the targeted poisoning is less effective (with 3--10\% success), but still, the defense reduces the success to 0--1\%. We observe that the two clean-label backdoor attacks show 1--7\% success regardless of the attacker's knowledge. We thus exclude these attacks from the subsequent evaluation. Using $\sigma$ greater than 0.5 increases the attack success by 0.5--2\%. But, this may not because of the successful poisoning attacks but due to the significant decrease in a model's accuracy (3-9\%). Increasing $\sigma$ may defeat a stronger attack but significantly decrease a model's accuracy. This result is interesting, as in our theoretical analysis shown in \S~\ref{sec:certified-defense}, our defense needs $\sigma$ of 0.25, while 0.1 is sufficient to defeat both the attacks.

\noindent
\textbf{Training from-scratch scenarios.}
Table~\ref{tbl:cifar10-denoise-linf16-scratch-all} shows our results with the same format as in Table~\ref{tbl:cifar10-denoise-linf16-transfer-learning-all}. We demonstrate that our defense mitigates both attacks at $\sigma=$ 0.25, aligned to what we've seen in our theoretical analysis. In Witches' Brew, the attack success decreases from 34--71\% to 3--10\%; against Sleeper Agent, the attack success decreases from 19--40\% to 7--19\%
Note that the success rate of typical backdoor attacks is $\sim$90\%; thus, the success below 19\% means they are ineffective.%
Note that to reduce their attack success to $\sim$10\%, our defense needs a minimum $\sigma$ of 0.25--0.5. It is also noteworthy that our defense reduces the backdooring success significantly while the certificate does not hold because test-time samples are altered with trigger patterns.

\begin{table*}[ht]
\centering
\caption{\textbf{Defense effectiveness in training from-scratch scenarios (CIFAR10).} We measure the accuracy and attack success of the models trained on the denoised training set. Each cell shows the accuracy in the parentheses and the attack success outside. Note that $^\dagger$ indicates the runs with $\sigma$=0.0, the same as our baseline that trains models without any denoising.
We use an ensemble of four models, and WB and BB stand for the white-box and the black-box attacks, respectively.}
\label{tbl:cifar10-denoise-linf16-scratch-all}
\vspace{-0.6em}
\adjustbox{max width=1.0\linewidth}{
\begin{tabular}{@{}lc | rrrrr | rrrrr@{}}
    \toprule
    \multicolumn{1}{c}{} & 
    & \multicolumn{5}{|c}{\textbf{Our defense against $\ell_2$ attacks at $\sigma$ (\%)}} 
    & \multicolumn{5}{|c}{\textbf{Our defense against $\ell_{\infty}$ attacks at $\sigma$ (\%)}} \\ \cmidrule(l){3-12} 
    \textbf{Poisoning attacks} & \textbf{Knowledge} 
    & $^\dagger$\textbf{0.0} & \textbf{0.1} & \textbf{0.25} & \textbf{0.5} & \textbf{1.0} 
    & $^\dagger$\textbf{0.0} & \textbf{0.1} & \textbf{0.25} & \textbf{0.5} & \textbf{1.0} 
    \\ \midrule \midrule
    Witches' Brew & \multirow{2}{*}{\rotatebox[origin=c]{90}{WB}} 
    & $^{(92.2)}$71.0 & $^{(86.4)}$54.0 & $^{(72.3)}$10.0 & $^{(46.7)}$11.0 & $^{(42.5)}$10.0
    & $^{(92.3)}$65.0 & $^{(86.5)}$9.0 & $^{(71.9)}$3.0 & $^{(46.0)}$9.0 & $^{(41.3)}$7.0
    \\
    Sleeper Agent Backdoor &  
    & $^{(92.4)}$40.5 & $^{(84.4)}$66.5 & $^{(71.8)}$19.5 & $^{(46.8)}$13.0 & $^{(39.9)}$10.5
    & $^{(92.4)}$35.0 & $^{(86.1)}$17.0 & $^{(73.0)}$8.0 & $^{(47.0)}$9.5 & $^{(39.7)}$10.0
    \\ \midrule
    Witches' Brew &  \multirow{2}{*}{\rotatebox[origin=c]{90}{BB}} 
    & $^{(90.1)}$45.5 & $^{(85.9)}$28.0 & $^{(75.5)}$4.0 & $^{(58.8)}$7.0 & $^{(49.0)}$10.0
    & $^{(90.0)}$33.5 & $^{(85.8)}$3.5 & $^{(75.5)}$2.5 & $^{(58.8)}$6.0 & $^{(48.7)}$6.5 \\
    Sleeper Agent Backdoor &  
    & $^{(90.0)}$39.5 & $^{(85.0)}$44.5 & $^{(75.1)}$14.5 & $^{(58.6)}$9.5 & $^{(49.0)}$8.0
    & $^{(90.0)}$18.5 & $^{(85.6)}$11.5 & $^{(75.5)}$7.0 & $^{(58.8)}$8.0 & $^{(48.4)}$8.0
    \\ \bottomrule
\end{tabular}
}
\vspace{-1.em}      %
\end{table*}

\subsection{Improving Model Utility}
\label{subsec:compare-to-existing-defense}

A shortcoming of certified defenses~\citep{DPA, DPA2, BagFlip} is that the %
robust model's accuracy is substantially lower than that of the undefended model. 
We also observe in %
\S\ref{subsec:our-defense-eval} that our defense, when we train a model from scratch on the data, denoised with large $\sigma$ values, is not free from the same issue, which hinders their deployment in practice.

In \S\ref{sec:certified-defense}, we show that our defense is agnostic to how we train a model, 
including how we \emph{initialize} a model's parameters. 
We leverage this property and minimize the utility loss by initializing model parameters in specific ways
before we train the model. %
We test two scenarios: a defender can use in-domain data or out-of-domain data to pre-train a model and use its parameters to initialize.
We evaluate %
them against the three most effective attacks %
shown in Table~\ref{tbl:cifar10-denoise-linf16-transfer-learning-all}.

\noindent
\textbf{Initializing a model using in-domain data.}
Our strategy here is an adaptation of warm-starting~\citep{ash2020warm}. We first train a model from scratch on the tampered training data for a few epochs to achieve high accuracy. It only needs 5--10 epochs in CIFAR10. We then apply our defense and continue training the model on the denoised training set. %

\begin{table}[ht]
\centering
\caption{%
\textbf{Improving the utility of defended models.} We show the accuracy of defended models and the attack success %
after employing our two strategies. %
The top three rows are the baselines from Table~\ref{tbl:cifar10-denoise-linf16-transfer-learning-all} and ~\ref{tbl:cifar10-denoise-linf16-scratch-all}, 
and the next two sets of three rows are our results. We highlight the cells showing the accuracy improvements in \textbf{bold}. ResNet18 models are used. BP, WB, and SA indicate Bullseye Polytope, Witches' Brew, and Sleeper Agent, respectively.}
\vspace{-0.6em}
\adjustbox{max width=\linewidth}{
\begin{threeparttable}
\begin{tabular}{@{}clrrrrr@{}}
    \toprule
    \multicolumn{1}{c}{\textbf{}} & \textbf{} & \multicolumn{5}{c}{\textbf{Our defense against $\ell_{\infty}$ attacks at $\sigma$ (\%)}} \\ \cmidrule(l){3-7} 
    \textbf{Att.} & \textbf{Initialization} & \textbf{$^{\dagger}$0.0} & \textbf{0.1} & \textbf{0.25} & \textbf{0.5} & \textbf{1.0} \\ \midrule \midrule
    BP & \multirow{3}{*}{\begin{tabular}[l]{@{}l@{}}N/A\\ (Baseline)\end{tabular}} 
    & $^{(93.5)}$100 & $^{(93.3)}$0.0 & $^{(92.7)}$0.0 & $^{(90.8)}$1.0 & $^{(87.6)}$0.0 \\ %
    WB &  & $^{(92.3)}$65.0 & $^{(86.5)}$9.0 & $^{(71.9)}$3.0 & $^{(46.0)}$9.0 & $^{(41.3)}$7.0 \\
    SA &  & $^{(92.4)}$34.0 & $^{(86.2)}$18.5 & $^{(73.0)}$7.9 & $^{(47.2)}$8.2 & $^{(40.1)}$11.5 \\ \midrule
    BP & \multirow{3}{*}{\begin{tabular}[l]{@{}l@{}}In-domain\\ (CIFAR10)\end{tabular}} & $^{(93.5)}$100 & $^{(93.3)}$0.0 & $^{(92.6)}$1.0 & $^{(\mathbf{90.9})}$1.0 & $^{(\mathbf{87.7})}$2.0 \\
    WB &  & $^{(92.3)}$65.0 & $^{(86.5)}$9.0 & $^{(\mathbf{72.7})}$5.0 & $^{(\mathbf{47.2})}$7.0 & $^{(38.6)}$8.0 \\ %
    SA &  & $^{(92.4)}$34.0 & $^{(86.2)}$18.1 & $^{(\mathbf{73.4})}$8.5 & $^{(\mathbf{47.8})}$9.9 & $^{(36.9)}$11.2 \\ \midrule
    BP & \multirow{3}{*}{\begin{tabular}[l]{@{}l@{}}$^\dagger$Out-of-domain\\ (ImageNet-21k)\end{tabular}} & 
                             $^{(85.2)}$4.0 & $^{(71.1)}$7.0 & $^{(66.4)}$5.0 & $^{(58.6)}$5.0 & $^{(48.7)}$6.0 \\  %
    WB &  & $^{(86.6)}$14.0 & $^{(84.1)}$4.0 & $^{(\mathbf{79.0})}$0.0 & $^{(\mathbf{67.7})}$5.0 & $^{(\mathbf{54.0})}$6.0 \\
    SA &  & $^{(92.3)}$35.0 & $^{(86.1)}$18.0 & $^{(72.9)}$8.5 & $^{(47.6)}$10.0 & $^{(37.4)}$10.0 \\ \bottomrule
\end{tabular}
\begin{tablenotes}
\item $^\dagger$ResNet18 in Torchvision library; only the latent space dimension differs. %
\end{tablenotes}
\end{threeparttable}
}
\label{tbl:cifar10-linf16-pretraining-trick}
\vspace{-1.8em}     %
\end{table}

Table~\ref{tbl:cifar10-linf16-pretraining-trick} summarizes our results. The middle three rows are the results of leveraging the in-domain data. We train the ResNet18 models on the tampered training set for 10 epochs and continue training on the denoised %
training data for the rest 30 epochs. Our first strategy (warm-starting) increases the accuracy of the defended models while defeating clean-label poisoning attacks. Under strong defense guarantees $\sigma\!>\!0.1$, the models have 0.5--2.2\% increased accuracy, compared to the baseline, while keeping the attack success $\sim$10\%. In $\sigma=$ 0.1, we achieve the same accuracy and defense successes. Our strategy can be potentially useful when a defender needs a stronger guarantee, such as against stronger clean-label poisoning future work will develop.

\noindent
\textbf{Using models pre-trained on out-of-domain data.}
Now, instead of running warm-starting on our end, we can also leverage ``warm" models available from the legitimate repositories. To evaluate, we take the ResNet18 model pre-trained on ImageNet-21k. %
We use this pre-trained model in two practical ways. We take the model as-is and train it on the denoised training data. In the second approach, we combine the first approach with the previous idea. We train the model on the tampered training data for a few epochs and then train this fine-tuned model on the denoised data.

The bottom three rows of Table~\ref{tbl:cifar10-linf16-pretraining-trick} are the results of using the warm model. We fine-tune the ResNet18 for 40 epochs on the denoised training data. %
In many cases, our second strategy can improve the accuracy of the defended models trained with strong defense guarantees ($\sigma\!>$ 0.1). These models achieve 6.5--13.2\% greater accuracy than the results shown in Table~\ref{tbl:cifar10-denoise-linf16-scratch-all} and~\ref{tbl:cifar10-denoise-linf16-transfer-learning-all}. In $\sigma\!=$ 1.0, the final accuracy of defended models has a negligible difference from the baseline. We are the first work to offer practical strategies to manage the utility-security trade-off in certified defense.

\subsection{Comparison to Existing Poisoning Defenses}
\label{subsec:comparison-to-prior-work}

We finally examine how our defense works better/worse than %
five %
existing defenses: %
k-NN~\citep{DeepKNN:20}, DP-SGD~\citep{DP:19, DP:20}, AT~\citep{AP:21}, FrieNDs~\citep{FrieNDs:22}, and ROE~\citep{rezaei2023run}. The ROE is a certified defense%
, and the other four are non-certified ones. %
We use %
$\ell_{\infty}$-norm of 16 in CIFAR10.

\begin{table}[h]
\centering
\caption{\textbf{Comparison of ours to Deep kNN.}
}
\label{tbl:knn-defense}
\vspace{-0.6em}
\adjustbox{max width=0.84\linewidth}{%
\begin{tabular}{@{}l|rrr|r@{}}
\toprule
\textbf{Poisoning attacks} & \textbf{TP} & \textbf{FP} & \textbf{Trained} & \textbf{Ours {(\small$\sigma\!=\!0.1$)}} \\ \midrule \midrule
Witches' Brew & 31 & 464 & $^{(86.3)}$\,\,\,3.0 & $^{(86.5)}$9.0 \\
Bullseye Polytope & 431 & 63 & $^{(93.7)}$25.0 & $^{(93.3)}$0.0 \\ \bottomrule
\end{tabular}
}
\vspace{-1.em}      %
\end{table}

\noindent
\textbf{Deep kNN} removes poisons from the training data by leveraging the k-nearest-neighbor (kNN) algorithm. They run kNN on the latent representations of the training data to identify potentially malicious samples (i.e., samples with labels different from the nearest ones) and remove them. We compare Deep kNN's effectiveness to our defense. We use the defense configurations that bring the best results in the original study, i.e., setting $k$ to 500. Table~\ref{tbl:defense-comparison} shows our results. Deep kNN fails to remove most poisoning samples (but it removes many benign samples!).

\begin{table}[ht]
\centering
\caption{\textbf{Comparison of ours to training-time defenses.}
}
\vspace{-0.6em}
\adjustbox{max width=\linewidth}{
    \begin{threeparttable}
    \begin{tabular}{@{}l|rrr|rr@{}}
    \toprule
     & \multicolumn{3}{c|}{\textbf{Heuristic defense}} & \multicolumn{2}{c}{\textbf{Certified defense}} \\ \midrule
    \textbf{Poisoning attacks} & \multicolumn{1}{c}{\textbf{DP-SGD}} & \multicolumn{1}{c}{\textbf{AT}} & \multicolumn{1}{c|}{\textbf{FrieNDs}} & \multicolumn{1}{c}{\textbf{ROE}} & \multicolumn{1}{c}{\textbf{Ours}} \\ \midrule \midrule
    Bullseye Polytope & $^{(93.9)}$7.0 & $^{(90.3)}$96.0 & $^{(90.2)}$10.0 & {\footnotesize $^{*}$N/A} & $^{(93.5)}$\,\,\,0.0 \\
    Witches' Brew & $^{(76.0)}$4.0 & $^{(66.5)}$\,\,\,2.0 & $^{(87.6)}$\,\,\,8.0 & $^{(70.2)}$10.0 & $^{(86.5)}$\,\,\,9.0 \\
    Sleeper Agent & $^{(74.8)}$5.0 & $^{(69.0)}$\,\,\,6.0 & $^{(87.4)}$\,14.0 & $^{(68.6)}$12.0 & $^{(86.0)}$17.0 \\ \bottomrule
    \end{tabular}
    \begin{tablenotes}
    \item $^{*}$ROE is incompatible with BP as ROE needs to train 250 models \emph{from scratch}.
    \end{tablenotes}
    \end{threeparttable}
}
\label{tbl:defense-comparison}
\vspace{-1.em}      %
\end{table}

\noindent
\textbf{DP-SGD.} \citet{DP:19} proposed a certified defense against poisoning attacks that leverage differential privacy (DP). DP makes a model less sensitive to a single-sample modification to the training data. But the guarantee is weak in practice; for example, the certificate is only valid in CIFAR10, when an adversary tampers one training sample. \citet{DP:20} later empirically shows that DP is still somewhat effective against poisoning. We compare our defense to the training with $(1.0, 0.05)$-DP%
, follow the prior work~\citep{lecuyer2019certified}.
In Table~\ref{tbl:defense-comparison}, %
DP significantly reduces the attack success to 0--16\%. Our defense achieves at most 10\% higher accuracy than DP. DP reduces the attack success slightly more against clean-label backdooring.

\noindent
\textbf{AT.} \citet{AP:21} adapts the adversarial training (AT) %
for defeating clean-label poisoning: in each mini-batch, instead of crafting adversarial examples, they synthesize poisoning samples and have a model to make correct classifications on them. %
While effective, the model could overfit a specific poisoning attack %
used in the adapted AT, leaving the model vulnerable to unseen poisoning attacks. We thus compare our defense with the original AT with the PGD-7 and $\varepsilon$ of 4%
, assuming that clean-label poisons %
are already adversarial examples. AT may not add more perturbations to %
them during training. %
In Table~\ref{tbl:defense-comparison}, %
our defense achieves higher accuracy than the robust models while making the attacks ineffective. Interestingly, AT cannot defeat the BP. %

\noindent
\textbf{FrieNDs.} \citet{FrieNDs:22} advances the idea of training robust models. Instead of adding random Gaussian noise to the data during training, they use ``friendly" noise, pre-computed by a defender, that minimizes the accuracy loss to a model they will train. When their defense is in action, they train a model on the tampered training data with two types of noise: the friendly noise computed before and a weak random noise sampled uniformly within a bound. %
Ours and FrieNDs greatly reduce the attack success while preserving the model's utility. The critical difference is that ours is a certified defense, while FrieNDs is not. A future work we leave is developing adaptive %
attacks against FrieNDs.

\noindent
\textbf{ROE.}
Certified defenses%
~\citep{steinhardt2017certified, DP:19, diakonikolas2019sever, DPA, gao2021learning, BagFlip, DPA2, rezaei2023run} are; however, %
most of these defenses showed their effectiveness in limited scenarios, such as in MINST-17~\citep{BagFlip}, a subset of MNIST containing only the samples of the digits 1 and 7. In consequence, they aren’t compatible with our clean-label poisoning settings. A recent defense by~\citet{rezaei2023run} demonstrates the certified robustness in practical settings; we thus compare our defense to their run-off-election (ROE). In Table~\ref{tbl:defense-comparison}, we empirically show that our defense is comparable to ROE while %
we preserve the model accuracy more.

\textbf{Note on the computational efficiency.}
To be practically effective, 
most training-time defenses require additional mechanisms we need to ``add" to the training algorithm.
AT %
craft adversarial examples or clean-label poisons~\citep{AP:21}. DP-SGD introduces noise into the gradients during training. FrieNDs pre-computes friendly noise and then optimally applies the noise during the training, and ROE requires training of 250 base classifiers to achieve the improved accuracy. In contrast, we only require to run forwards \emph{once} with an off-the-shelf diffusion model.

\section{Discussion}
\label{sec:discussion}

Our work offers a new perspective on the cat-and-mouse game played in clean-label poisoning research.

\noindent
\textbf{Initial claims from the cat:}
Clean label poisoning attacks, with human-imperceptible perturbations that keeping the original labels, are an attractive option to inject malicious behaviors into models. Even if a victim carefully curates the training data, %
they remain effective by evading \emph{filtering} defenses that rely on statistical signatures distinguishing clean samples from poisons~\citep{RONI, tRONI:18, DeepKNN:20}.

\textbf{The mouse's counterarguments:} Yet, as seen in research on the adversarial robustness, such small perturbations can be \emph{brittle}~\citep{carlini2022certified}: by adding a large quantity of noise and then denoising the images, the adversarial perturbations become nearly removed---and so clean-label perturbations are also removed. In prior work (and in this work as well), we leverage this observation and propose both certified and non-certified heuristic defenses.

\noindent
\textbf{Additional arguments from the cat:}
A counterargument from the clean-label attack's side was that those defenses inevitably compromise a model's performance---a fact corroborated by the prior work on the adversarial robustness~\citep{tsipras2018robustness, zhang2019theoretically}. Similarly, defenses could greatly reduce the poisoning attack success, but at the same time, they decrease the %
accuracy of defended models, often tipping the scales in favor of the adversary. If a defense yields a CIFAR-10 model with 60--70\% utility, what is the point of having such models in practice? Indeed, our own certified models require degraded utility to achieve %
certified predictions. Other defenses, e.g., DP-SGD~\citep{DP:19, DP:20}, are computationally demanding, increasing the training time %
an order of magnitude.

This %
game need not be as pessimistic as the cat portrays.

By leveraging an off-the-shelf %
DDPMs~\citep{ho2020denoising, nichol2021improved}, we can purify the training data, possibly containing malicious samples, \emph{offline} and render six clean-label attacks ineffective. With a weak provable guarantee against the training data contamination ($\ell_{\infty}$-norm of 2), the majority of attacks reach to 0\% success. A few recent attacks exhibit 0--10\% success rate%
---a rate comparable to random misclassification of test-time inputs.

In contrast to the cat's counterargument, %
we can also address the accuracy degradation problem---as long as we do not need certified robustness. Existing defenses that work to defend against poisoning require applying a particular training algorithm designed by the defender~\citep{DP:20, AP:21, FrieNDs:22, DPA, DPA2, BagFlip}. By decoupling the training process and the defense, we open a new opportunity %
to develop training techniques that can enhance the accuracy of the defended models while keeping the robustness. Our results propose two effective heuristic algorithms in improving defended models' utility, especially when the models are under strong defense guarantees (i.e., $||\epsilon||_{\infty} > 2$).

\newpage
\section*{Impact Statement}   %

Our results have no immediate negative consequences.
Below we discuss a few implications of our work for how
the field of poisoning may progress as a result of our methods.

First, our results do not imply that clean-label poisoning attacks are an ineffective threat and that we should stop studying. Instead, the results shows that clean-label poisoning has been studied in adversary-friendly settings. Current attacks are \emph{weak}: an adversary crafts poisons bounded to the perturbations of 16-norm, yet they become ineffective with our $2$-norm bounded defense. This implies the current attacks are \emph{not-fully} exploiting the 16-norm perturbation bound and the need to study the worst-case poisoning adversary who remains effective under our defense.

In addition, when studying stronger clean-label poisoning attacks, we encourage future work to include our defense as a strong baseline. We show in \S\ref{sec:evaluation} that the defense renders the state-of-the-art clean-label attacks ineffective and is, by far, the most effective defense among the existing ones. We believe that strong attacks future work will develop should at least survive against this defense. Our defense is also designed to run without any additional training costs---with an off-the-shelf diffusion model, one can easily denoise the tampered training data and run the training of a model on it.

It may be possible for future work to tighten the certified robustness our defense offers further. Making an analogy to the privacy-utility trade-off, prior work studies a series of methods that can improve the utility of differentially-private models with the same worst-case private leakage~\citep{tramer2020differentially, papernot2021tempered, ye2022differentially}. Similarly, in the context of clean-label poisoning, a promising direction for future work will be to find methods to achieve a better trade-off between the certified robustness to the clean-label attacks and the model utility.

{
    \bibliographystyle{icml2025}
    \bibliography{bib/security}
}

\newpage
\appendix
\onecolumn
\newpage
\appendix
\onecolumn

\section{Experimental Setup in Detail}
\label{appendix:detailed-configurations}

We adapt the poisoning benchmark from the prior work~\citep{PoisoningBenchmark}. Most work on clean-label poisoning attacks uses this benchmark for showcasing their attack success. The open-sourced implementation of Sleeper Agent is incompatible with this framework; thus, we run this attack separately. We implement the attacks to use a comparable $\ell_2$-norm bound in the poison-crafting process (see \S\ref{appendix:norm-conversion}). We implemented our benchmarking framework in Python v3.7\footnote{\href{https://python.org}{https://www.python.org/}} and PyTorch v1.13\footnote{\href{https://pytorch.org}{https://pytorch.org}}. We use the exact attack configurations and training hyper-parameters that the original study employs. 

In contrast to the original work, we made two major differences. We first increase the perturbations bounded to $||\epsilon||_{\infty}\!=\!16$ as the 8-pixel bound attacks do not result in a high attack success rate. Defeating 8-pixel bounded attacks is trivial for any poisoning defenses. Second, we do not use their fine-tuning subset, which only contains 2500 training samples and 25--50 poisons. It (as shown in the next subsections) leads to 60--70\% accuracy on the clean CIFAR10 test-set, significantly lower than 80--90\%, which can be trivially obtained with any models and training configurations. If a model trained on the contaminated training data misclassifies a target, it could be a mistake caused by a poorly performing model.

We run a comprehensive evaluation. We run 7 attacks; for each attack, we run 100 times of crafting and training/fine-tuning a model. We also examine 5 different denoising factor $\sigma$. In total, we ran 3500 poisoning attacks. For the most successful three attacks, we run 6 different certified and non-certified defenses over 100 poisoning runs. The three training-time defenses (DP-SGD, AT, FrieNDs) require 100 trainings of a model for each, and the certified defense (RoE) requires 100$\times$250 trainings in total. To accommodate this computational overhead, we use two machines, each equipped with 8 Nvidia GPUs. Crafting a poisoning set takes approximately one hour and training takes 30 minutes on a single GPU.

\section{Translating $l_{\infty}$-bound into $l_2$-bound}
\label{appendix:norm-conversion}

The certification we offer in Sec~\ref{sec:certified-defense} is defined in $l_2$-norm, while most poisoning attacks work with $l_{\infty}$-norm, e.g., of 8 or 16. We therefore convert these $l_{\infty}$-bounds into $l_2$-bounds. We assume the worst-case perturbation in the $l_{\infty}$-space that changes every pixel location of an image by 16 pixels, compute the $l_2$-norm of that perturbation as follows, and use it in \S\ref{sec:evaluation}:
\begin{itemize}[noitemsep, topsep=0.em]
    \item $l_{\infty}$-bound of 8 pixels: $\sqrt{3 * 32 * 32 * (8/255)^2} = 1.74$
    \item $l_{\infty}$-bound of 16 pixels: $\sqrt{3 * 32 * 32 * (16/255)^2} = 3.48$
\end{itemize}
Oftentimes, the $\ell_2$ attacks with the comparable bound 3.48 do not lead to comparable attack success. We, therefore, especially for the poisoning attacks in the transfer-learning scenarios, increase the bound to 6.43. We note that most existing attacks either do not implement $\ell_2$ bounds or are unbounded~\citep{Frogs}.

\section{Proof for Our Provable Guarantee}
\label{appendix:certification-proof}

This section provides the proof of our provable robustness
against clean-label poisoning attacks within the $\ell_2$-norm.

\textbf{Studies on certified robustness.}
We first note that certified robustness 
has been explored for other types of attacks, 
such as adversarial examples~\citep{cohen2019certified, denoisedsmoothing, carlini2022certified}
and backdoor attacks~\cite{RAB},
which assumes different threat models 
compared to clean-label poisoning.
\citet{cohen2019certified} 
(certified defenses against adversarial examples)
operate under the assumption that
the target classifier remains \emph{stationary},
whereas clean-label poisoning involves the victim
``trains" the classifier on the compromised data.
\citet{RAB} presents the provable guarantee against backdooring,
where classifiers are trained on datasets 
containing backdoor poisons.
But, backdoor attacks differ substantially
as they assume that the adversary can manipulate test-time inputs
with specific trigger patterns.

\textbf{Comparison to the prior studies.}
Our provable guarantee can be viewed as 
a special case of the backdoor robustness framework
proposed by~\citet{RAB}, where the backdoor adversary
applies zero perturbations ($\varepsilon\!=\!0$) to test-time inputs.
We also highlight that our use of diffusion denoising
as a randomized smoothing mechanism enhances 
classification accuracy by 1.8$\times$ compared to
the results reported by~\citep{RAB}.
Moreover, to achieve the same level of certified accuracy,
we need 20$\times$ less computational cost 
compared to prior studies~\cite{cohen2019certified, RAB}.

\textbf{Proof.}
On a high-level, we first review the robustness guarantee 
provided by~\cite{cohen2019certified} to adversarial examples,
specifically for $\ell_2$-norm perturbations in the input space.
Next, we elaborate on how~\citet{RAB} extend the previous robustness%
---achieved through training with randomized smoothing.
By smoothing the training data first and then training a classifier on it,
\citet{RAB} establish a provable guarantee against backdooring%
---an attack exploits data poisoning.
We finally reduce the robustness guarantee against backdooring
to a special case where no input perturbations occur at test-time,
thereby providing a provable guarantee against clean-label poisoning.

\begin{theorem}[Robustness guarantee by~\citet{cohen2019certified}]
\label{thm:cohen}
    Let $f: \mathbb{R}^{d}\rightarrow \mathcal{Y}$ be any deterministic (or random) function,
    $\varepsilon \sim \mathcal{N}(0, \sigma^2, I)$, and
    $g$ to be the smoothed classifier robust under $\varepsilon$.
    Suppose $c_A\in\mathcal{Y}$ and $p_A, p_B \in [0, 1]$ satisfy:
    \begin{align*}
        \mathbb{P}(f(x+\varepsilon)=c_A) \geq p_A > p_B
        \geq \underset{c \neq c_A}{\max}\mathbb{P}(f(x+\varepsilon)=c)
    \end{align*}
    Then $g(x+\delta)=c_A$ for all $||\delta||_2 < R$, where
    \begin{align*}
        R = \frac{\sigma}{2}\big(\Phi^{-1}(p_A) - \Phi^{-1}(p_B)\big)
    \end{align*}
\end{theorem}

Theorem~\ref{thm:cohen} establishes that 
the smoothed classifier $g$
is robust around a test-time input $x$
within the $\ell_2$-norm radius of $R$ as specified above.
Here $\Phi$ denotes the inverse of the standard Gaussian CDF.
For the proof of this theorem,
we refer readers to Appendix A
of the original study by~\citet{cohen2019certified}.

The previous approach by~\citet{cohen2019certified}
does not impose specific restrictions on 
\emph{how} the smoothed classifier $g$ is constructed,
A common method they employ is to 
apply additive Gaussian noise to the training data 
\emph{at each mini-batch during training}.
However, this method is not compatible with our setting
where we train a classifier \emph{on the smoothed training data}.
To address this challenge, 
we adopt the approach proposed by~\citet{RAB},
which relies on a smoothed classifier $g$
obtained through a method that aligns with our framework,
defined as follows:

\begin{definition}[Smoothed classifier considered by~\citet{RAB}]
\label{def:smoothed-classifier-rab}
    Suppose a base classifier $h$, defined as:
    \begin{align*}
        h(x, \mathcal{D}) = \argmax_{y}p(y|x, \mathcal{D}),
    \end{align*}
    where $p$ is learned from the training data $\mathcal{D}$
    and to maximize the conditional probability distribution over labels $y$.
    Then the smoothed classifier $q$ is defined by:
    \begin{align*}
        q(y|x, \mathcal{D}) = \mathbb{P}_{X, D}(h(x+X, \mathcal{D}+D)=y),
    \end{align*}
    where $X\sim\mathbb{P}_X$ represents random variables,
    and $D\sim\mathbb{P}_D$ denotes smoothing distributions.
    $D$ is a collection of $n$ i.i.d random variables,
    composed of $D^{i}$s, where each $D^{i}$ will be added to a training samples in $\mathcal{D}$.
    $\mathbb{P}_X$ is the noise distribution for test-time inputs,
    and $\mathbb{P}_D$ is the noise distribution for the training data.
    The final, smoothed classifier then assigns 
    the most likely class to an instance $x$ under $q$, so that:
    \begin{align*}
        g(x,\mathcal{D}) = \argmax_{y} q(y|x,\mathcal{D})
    \end{align*}
\end{definition}

Note that under this formulation, 
the smoothed classifier can defend against adversarial examples
by setting the training set noise to zero, $D \equiv 0$.
Similarly, setting the test-time noise to zero, $X \equiv 0$,
enables the classifier to defend against clean-label poisoning.
We will first follow the robustness framework by~\citet{RAB}
and then show the special case, $X \equiv 0$,
establishing our guarantee against clean-label poisoning.

\begin{theorem}[Robustness guarantee by~\citet{RAB}]
\label{thm:rab}
    Suppose $q$ is the smoothed classifier defined in Def.~\ref{def:smoothed-classifier-rab}
    with smoothing distribution $Z:=(X, D)$ with $X$ taking values in $\mathbb{R}^{d}$
    and $D$ being a collection of $n$ independent $\mathbb{R}^{d}$-valued random variables,
    $D=(D^1, D^2, ..., D^n)$.
    Let $\Omega_x \in \mathbb{R}^d$ and 
    $\Delta := (\delta^1, \delta^2, ..., \delta^n)$ 
    for backdoor patterns $\delta^i \in \mathbb{R}^d$.
    Let $y_A \in C$ and $p_A, p_B \in [0, 1]$
    such that $y_A = g(x, D)$ and
    \begin{align*}
        q(y_A|x, \mathcal{D}) \geq p_A > p_B \geq 
        \underset{y \neq y_A}{\max}q(y|x, \mathcal{D}).
    \end{align*}
    If the optimal type II errors, for testing the null $Z \sim \mathbb{P}_0$
    against the alternative $Z + (\Omega_x, \Delta) \sim \mathbb{P}_1$, satisfy
    \begin{align*}
        \beta^{*}(1-p_Z; \mathbb{P}_0, \mathbb{P}_1) +
        \beta^{*}(p_B; \mathbb{P}_0, \mathbb{P}_1) > 1,
    \end{align*}
    then it is guaranteed that $y_A = \argmax_{y}q(y|x+\Omega_x, \mathcal{D}+\Delta)$.
\end{theorem}

Theorem~\ref{thm:rab} by~\citet{RAB}
establishes that the smoothed classifier $q$,
as defined in Definition~\ref{def:smoothed-classifier-rab},
trained on data containing backdoor patterns $\Delta$,
is robust around a test-time input $x$
with the perturbations in $\Omega_x$.
For the proof of this theorem,
we refer readers to Appendix A.1
of the original study by~\citet{RAB}.

\begin{corollary}[$\ell_2$-robustness achieved by Gaussian smoothing in~\citet{RAB}]
\label{coll:l2-rab}
    Let $\Delta = (\delta_1, \delta_2, ..., \delta_n$),
    $\Omega_x$ be $\mathbb{R}^d$-valued backdoor patterns,
    and $\mathcal{D}$ be a training dataset.
    Suppose that for each $i$,
    the smoothing noise on the training features is 
    $D^i \stackrel{iid}\sim \mathcal{N}(0, \sigma^2 \mathbbm{1}_d)$.
    Let $y_A \in C$ such that $y_A=q(x+\Omega_x, \mathcal{D}+\Delta)$
    with class probabilities satisfying:
    \begin{align*}
        q(y_A|x+\Omega_x, \mathcal{D}+\Delta) \geq p_A >
        p_B \geq \underset{y \neq y_A}{\max} q(y|x+\Omega_x, \mathcal{D}+\Delta)
    \end{align*}
    Then if the backdoor patterns on the training samples are bounded by
    \begin{align*}
        \sqrt{\sum_{i=1}^{n} ||\delta_i||_2^2} <
        \frac{\sigma}{2}\big(\Phi^{-1}(p_A) - \Phi^{-1}(p_B)\big),
    \end{align*}
    it is guaranteed $y_A = q(x+\Omega_x, \mathcal{D}) = q(x+\Omega_x, \mathcal{D}+\Delta)$.
\end{corollary}

Now we can simplify Corollary~\ref{coll:l2-rab}
by assuming an adversary compromises $r \leq n$ training samples
with the largest noise pattern's $\ell_2$-norm bound, $||\delta||_2$:
\begin{align*}
    ||\delta||_2 < \frac{\sigma}{2\sqrt{r}}
        (\Phi^{-1}(p_A) - \Phi^{-1}(p_B)\big),
\end{align*}

A typical choice of $r$ in clean-label poisoning 
is $\sim$10\% of the training data.
But when computing our certificate bound in \S\ref{sec:certified-defense},
we consider the worst-case adversary 
capable of compromising the entire training set.
In this case, $r\!=\!n$, which becomes:
\begin{align*}
    ||\delta||_2 < \frac{\sigma}{2\sqrt{\textcolor{blue}{n}}}
        (\Phi^{-1}(p_A) - \Phi^{-1}(p_B)\big) 
                \leq \frac{\sigma}{2\sqrt{r}}
        (\Phi^{-1}(p_A) - \Phi^{-1}(p_B)\big),
\end{align*}

Under this formulation, 
we observe that against the strongest adversary 
who can compromise the entire training data $n$,
the robustness bound is much smaller compared to
that against the realistic attacker capable of 
compromising only 10\% of the training data.

\textbf{Certifying the robustness with multiple smoothed classifiers.}
In our certification process, a defender trains $n$ models%
\footnote{Note that this $n$ is different from $n$ we use for denoting the cardinality of the training set.}
that return the prediction of $x$, the target test-time sample, as $y_A$.
Then the lower bound of $p_A$ becomes $\alpha^{1/n}$,
with the probability of at least $1-\alpha$.
Plugging this into the above equation yields:
\begin{align*}
    ||\delta||_2 < \frac{\sigma}{\sqrt{n}} (\Phi^{-1}(p_A)\big)
\end{align*}

\begin{corollary}[Our robustness guarantee against clean-label poisoning.]
    Consider the smoothed classifier $q$ from Theorem~\ref{thm:rab},
    with no test-time input perturbation $X$ and 
    the smoothing mechanism being diffusion denoising,
    which applies a Gaussian process to add and remove noise from the input.
    Under this setting,
    Theorem~\ref{thm:rab} provides the same robustness guarantee 
    for clean-label poisoning attacks.
\end{corollary}

\begin{proof}
    By setting $X \equiv 0$ in Theorem~\ref{thm:rab} (and therefore $\Omega_x$),
    the smoothed classifier will output the same label prediction $y_A$
    within the same $\ell_2$-bound for the perturbations applied to the training data of size $n$.
\end{proof}

We also note that~\citet{RAB} formally demonstrate 
how to derive a robustness guarantee for $\ell_{\infty}$-norm perturbations
using uniform smoothing.
Because our work employs off-the-shelf diffusion denoising models
that rely on Gaussian random noise,
we do not provide a provable guarantee for $\ell_{\infty}$-based clean-label poisoning.
(but we do empirically show substantial effectiveness against $\ell_{\infty}$-based attacks!)
We believe that extending our work 
to provide formal guarantees for $\ell_{\infty}$-based attacks
would be straightforward and thus leave this for future work.

\section{More Experimental Results}
\label{appendix:additional-experiments}

Here we include additional experimental results.

\subsection{Using the Evaluation Setup by~\citet{PoisoningBenchmark}}
\label{subsec:what-if-in-their-setup}

We examine the weaker adversary whose perturbation is bounded to $||\delta||_{\infty} = 8$.
We also employ the exact setup as the prior work~\citep{PoisoningBenchmark},
where we craft 25 poisons on a ResNet18 pre-trained on CIFAR-100
and fine-tune the model on a subset of the CIFAR-10 training data. 
The subset contains the first 250 images per class (2.5k samples in total).

\begin{table}[ht]
\centering
\caption{\textbf{Diffusion denoising against clean-label poisoning (CIFAR10).} We denoise the $l_{\infty}$-norm of 8 perturbations added by five poisoning attacks by running a single-step stable diffusion on the entire training set. In each cell, we show the average attack success over 100 runs and the average accuracy of models trained on the denoised data in the parenthesis. Note that $^\dagger$ indicates the runs with $\sigma$=0.0, the same as our baseline that trains models without any denoising.}
\vspace{0.2em}
\adjustbox{max width=\linewidth}{
\begin{tabular}{@{}lcrrrrr@{}}
    \toprule
    \multicolumn{1}{c}{} &  & \multicolumn{5}{c}{\textbf{Our defense against $\ell_{\infty}$ attacks at $\sigma$ (\%)}} \\ \cmidrule(l){3-7} 
    \textbf{Poisoning attacks} & \textbf{Scenario} & $^\dagger$\textbf{0.0} & \textbf{0.5} & \textbf{1.0} & \textbf{1.5} & \textbf{2.0} \\ \midrule \midrule
    Poison Frog!~\citep{Frogs} & \multirow{5}{*}{\rotatebox[origin=c]{90}{White-box}} & $^{(69.8)}$13.0 & $^{(55.9)}$8.0 & $^{(43.1)}$4.0 & $^{(34.7)}$9.0 & $^{(26.9)}$11.0 \\
    Convex Polytope~\citep{CP} &  & $^{(69.8)}$24.0 & $^{(53.8)}$5.0 & $^{(42.6)}$9.0 & $^{(35.0)}$8.0 & $^{(27.7)}$6.0 \\
    Bullseye Polytope~\citep{BP} &  & $^{(69.4)}$100.0 & $^{(55.8)}$10.0 & $^{(42.8)}$8.0 & $^{(35.0)}$9.0 & $^{(27.4)}$9.0 \\
    Label-consistent Backdoor~\citep{LC} &  & $^{(69.8)}$2.0 & $^{(55.9)}$3.0 & $^{(42.7)}$5.0 & $^{(34.9)}$8.0 & $^{(27.4)}$12.0 \\      %
    Hidden Trigger Backdoor~\citep{HT} &  & $^{(69.8)}$5.0 & $^{(55.9)}$3.0 & $^{(42.7)}$10.0 & $^{(35.2)}$9.0 & $^{(27.3)}$9.0 \\ \midrule        %
    Poison Frog!~\citep{Frogs} & \multirow{5}{*}{\rotatebox[origin=c]{90}{Black-box}} & $^{(67.9)}$7.0 & $^{(53.8)}$6.0 & $^{(43.3)}$2.5 & $^{(35.3)}$8.0 & $^{(28.8)}$6.5 \\
    Convex Polytope~\citep{CP} &  & $^{(67.9)}$4.0 & $^{(53.7)}$3.0 & $^{(43.3)}$3.5 & $^{(35.2)}$3.0 & $^{(28.8)}$8.0 \\
    Bullseye Polytope~\citep{BP} &  & $^{(67.7)}$8.0 & $^{(53.8)}$17.5 & $^{(43.2)}$16.5 & $^{(35.3)}$7.5 & $^{(28.6)}$8.5 \\
    Label-consistent Backdoor~\citep{LC} &  & $^{(67.9)}$3.5 & $^{(53.8)}$2.0 & $^{(43.5)}$4.5 & $^{(35.2)}$2.5 & $^{(29.0)}$8.0 \\      %
    Hidden Trigger Backdoor~\citep{HT} &  & $^{(67.9)}$7.5 & $^{(53.7)}$2.0 & $^{(43.3)}$7.0 & $^{(35.2)}$10.5 & $^{(28.7)}$8.5 \\ \bottomrule       %
\end{tabular}
}
\label{tbl:cifar10-denoise-linf8}
\end{table}

Table~\ref{tbl:cifar10-denoise-linf8} summarizes our results against clean-label poisoning attacks with $l_{\infty}$-norm of 8. We examine five poisoning attacks in the white-box and black-box scenarios. We focus on the attacks against transfer-learning as the specific data splits the prior work~\citep{PoisoningBenchmark} uses are only compatible with them. Poisoning attacks against training from scratch scenarios use the entire CIFAR10 as ours, so we don't need to duplicate the experiments. In the white-box setting, we fine-tune the CIFAR-100 ResNet18 model (that we use for crafting poisons) on the poisoned training set. In the black-box setting, we fine-tune different models (VGG16 and MobileNetV2 models pre-trained on CIFAR-100) on the same poisoned training set. We use $\sigma$ in $\{0.5, 1.0, 1.5, 2.0\}$ for our single-step DDPM. $\sigma$=0.0 is the same as no defense.

\textbf{Results.}
Diffusion denoising significantly reduces the poisoning attack success. The most successful attack, Bullseye Polytope in the white-box setting, achieves the attack success of 100\% in $l_{\infty}$-norm of 8 pixels, but denoising with $\sigma$ of 0.5 can reduce their success to 10\%. Our defense reduces the attack success of Poison Frog! and Convex Polytope from 13-34\% to 2-8\% at $\sigma= 0.5$. The two backdoor attacks (label-consistent and hidden trigger) exploiting clean-label poisoning are not successful in the benchmark setup (their success rate ranges from 2-7.5\%). We thus could not quantify our defense's effectiveness against these backdoor attacks. Note that \citet{PoisoningBenchmark} also showed these backdoor attacks ineffective, and our finding is consistent with their results. We exclude them from our Tiny ImageNet results in \S\ref{subsec:tiny-imagenet}.

We also observe that the increased $\sigma$ (strong denoising) can significantly reduce the utility of a model trained on the denoised training data. Note that since the setup uses only 2.5k CIFAR10 samples for fine-tuning, and in consequence, the model's utility is already $\sim$70\% at most, much lower than our setup. As we increase the $\sigma$ from 0.5 to 2.0, the fine-tuned model's accuracy leads to 56\% to 27\%. However, we show that with the small $\sigma$, our diffusion denoising can reduce the attack success significantly. We also show in our evaluation (\S\ref{sec:evaluation}) that we achieve a high model's utility while keeping the same $\sigma$. We attribute the increased utility to recent model architectures, such as VisionTransformers, or to pre-training a model on a larger data corpus. We leave further investigation for future work.

Moreover, in a few cases, the poisoning success increases from 2--7\% to 10--13\% as we increase $\sigma$. We attribute this increase not to the attack being successful with a high $\sigma$ but to the poor performance of a model. For example, the accuracy of a model with $\sigma=2.0$ is $\sim$27\%, meaning that four out of five targets in a class can be misclassified.

\subsection{Tiny ImageNet Results}
\label{subsec:tiny-imagenet}

\begin{table}[ht]
\centering
\caption{\textbf{Diffusion denoising against clean-label poisoning (Tiny ImageNet).} We consider four attacks with the $\ell_{\infty}$-norm of 16 perturbation bound. In each cell, the attack success and accuracy of models trained on the denoised data in the parenthesis on average over 100 runs. Note that $^\dagger$ indicates the runs with $\sigma$=0.0, representing no-defense scenario.}
\vspace{0.2em}
\adjustbox{max width=\linewidth}{
\begin{tabular}{@{}lcrrrrr@{}}
    \toprule
    \multicolumn{1}{c}{} &  & \multicolumn{5}{c}{\textbf{Our defense against $\ell_{\infty}$ attacks at $\sigma$ (\%)}} \\ \cmidrule(l){3-7} 
    \textbf{Poisoning attacks} & \textbf{Scenario} & $^\dagger$\textbf{0.0} & \textbf{0.1} & \textbf{0.25} & \textbf{0.5} & \textbf{1.0} \\ \midrule \midrule
    Poison Frog!~\citep{Frogs} & \multirow{3}{*}{\rotatebox[origin=c]{90}{WB}} 
    & $^{(58.9)}$79.0 & $^{(56.3)}$3.0 & $^{(53.3)}$2.0 & $^{(46.5)}$2.0 & $^{(32.0)}$1.0 \\
    Convex Polytope~\citep{CP} &  & $^{(58.9)}$95.0 & $^{(56.3)}$14.6 & %
    - & %
    - & %
    - \\
    Bullseye Polytope~\citep{BP} &  & $^{(58.8)}$100.0 & $^{(57.4)}$100.0 & %
    - & $^{(48.7)}$0.0 & %
    - \\ \midrule %
    Poison Frog!~\citep{Frogs} & \multirow{3}{*}{\rotatebox[origin=c]{90}{BB}} 
    & $^{(58.0)}$5.0 & $^{(51.9)}$0.0 & $^{(46.5)}$0.0 & $^{(36.0)}$1.0 & $^{(22.6)}$1.0 \\
    Convex Polytope~\citep{CP} &  & $^{(58.1)}$0.0 & $^{(51.9)}$0.0 & %
    - & %
    - & %
    - \\
    Bullseye Polytope~\citep{BP} &  & $^{(57.9)}$9.0 & $^{(59.1)}$33.3 & %
    - & $^{(49.8)}$16.7 & %
    - \\ \bottomrule %
\end{tabular}
}
\label{tbl:tiny-denoise-linf16-all}
\end{table}

We also ran our experiments with Tiny ImageNet to examine whether our findings are consistent across different datasets. We assume the same adversary who can add perturbations bounded to $\ell_{\infty}$-norm of 16 pixels. For the three attacks in transfer-learning scenarios, we craft 250 poisoning samples on a VGG16 pret-trained on the same dataset and fine-tune the model on the tampered training set. %
We do not evaluate backdoor attacks because they are either ineffective (label-consistent and hidden trigger) or the original study did not employ Tiny ImageNet (Sleeper Agent).

\textbf{Results.}
Table~\ref{tbl:tiny-denoise-linf16-all} shows our results. Our results are consistent with what we observe in CIFAR10. Note that the poisoning attacks are more successful against Tiny ImageNet, and we attribute the success to the learning complexity of Tiny ImageNet over CIFAR10: Tiny ImageNet has 200 classes and 100k training samples. Prior work~\citep{PoisoningBenchmark} had the same observation. Our defense significantly reduces the attack success. Bullseye Polytope in the white-box setting, achieves the attack success of 100\% without our defense, but denoising with $\sigma$ of 0.25 reduces their success to 10--15\%. We find the defense reduces the success of Poison Frog! and Convex Polytope from 79--95\% to 3--15\% at $\sigma= 0.1$. %
We also observe that the increase in $\sigma$ significantly degrades the model utility., e.g., $\sigma=$1.0 loses the accuracy by 27\%.

\section{(Denoised) Poisoning samples}
\label{appendix:poisoning-samples}

\begin{table}[ht]
\centering
\caption{\textbf{Visualize poisoning samples.} We, for the CIFAR10 training data, display the poisoning samples crafted by different clean-label poisoning attacks ($\ell_{\infty}$-norm of 16). We also show how the perturbations are denoised with difference $\sigma$ values in $\{0.1, 0.25, 0.5, 1.0\}$. $\sigma=$ 0.1 yields to ineffective poisons. $\sigma=$ 0.0 means we do not denoise the poisons.}
\adjustbox{max width=1.0\linewidth}{
\begin{tabular}{@{}cccccccl@{}}
    \toprule
    \textbf{FC} & \textbf{CP} & \textbf{BP} & \textbf{WB} & \textbf{HTBD} & \textbf{CLBD} & \textbf{SA} & \multicolumn{1}{c}{$\mathbf{\sigma}$} \\ \midrule \midrule
    \includegraphics[scale=2.]{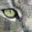} 
        & \includegraphics[scale=2.]{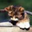} 
        & \includegraphics[scale=2.]{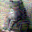} 
        & \includegraphics[scale=2.]{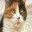}
        & \includegraphics[scale=2.]{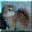}
        & \includegraphics[scale=2.]{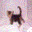}
        & \includegraphics[scale=2.]{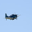} 
        & \textbf{0.0} \\
    \includegraphics[scale=2.]{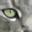} 
        & \includegraphics[scale=2.]{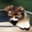} 
        & \includegraphics[scale=2.]{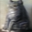} 
        & \includegraphics[scale=2.]{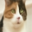}
        & \includegraphics[scale=2.]{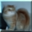}
        & \includegraphics[scale=2.]{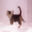}
        & \includegraphics[scale=2.]{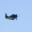} 
        & \textbf{0.1} \\
    \includegraphics[scale=2.]{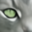} 
        & \includegraphics[scale=2.]{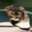} 
        & \includegraphics[scale=2.]{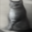} 
        & \includegraphics[scale=2.]{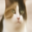}
        & \includegraphics[scale=2.]{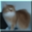}
        & \includegraphics[scale=2.]{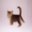}
        & \includegraphics[scale=2.]{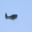} 
        & \textbf{0.25} \\
    \includegraphics[scale=2.]{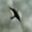} 
        & \includegraphics[scale=2.]{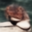} 
        & \includegraphics[scale=2.]{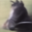} 
        & \includegraphics[scale=2.]{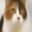}
        & \includegraphics[scale=2.]{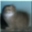} 
        & \includegraphics[scale=2.]{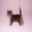} 
        & \includegraphics[scale=2.]{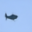} 
        & \textbf{0.5} \\
    \includegraphics[scale=2.]{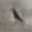} 
        & \includegraphics[scale=2.]{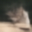} 
        & \includegraphics[scale=2.]{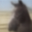} 
        & \includegraphics[scale=2.]{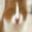} 
        & \includegraphics[scale=2.]{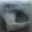} 
        & \includegraphics[scale=2.]{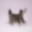} 
        & \includegraphics[scale=2.]{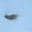} 
        & \textbf{1.0} \\ \bottomrule
\end{tabular}
}
\end{table}

\end{document}